\documentclass[journal,comsoc]{IEEEtran}
\usepackage[T1]{fontenc}
\usepackage{algorithm,algorithmic,amsbsy,amsmath,amssymb,epsfig,bbm,mathrsfs,multirow,amsthm}
\usepackage{subcaption,float}
\usepackage[hidelinks]{hyperref}

\newtheorem{theorem}{Theorem}
\begin{document}

\title{Optimal Time Scheduling for Wireless-Powered Backscatter Communication Networks}

\author{Nguyen Van Huynh$^1$, Dinh Thai Hoang$^1$, Dusit Niyato$^1$, Ping Wang$^1$, and Dong In Kim$^2$ \\
	$^1$ School of Computer Science and Engineering, Nanyang Technological University (NTU), Singapore\\
	$^2$ School of Information and Communication Engineering, Sungkyunkwan University (SKKU), Korea	\vspace{-5mm}}
\maketitle
\begin{abstract}
This letter introduces a novel wireless-powered backscatter communication system which allows sensors to utilize RF signals transmitted from a dedicated RF energy source to transmit data. In the proposed system, when the RF energy source transmits RF signals, the sensors are able to backscatter the RF signals to transmit date to the gateway and/or harvest energy from the RF signals for their operations. By integrating backscattering and energy harvesting techniques, we can optimize the network throughput of the system. In particular, we first formulate the time scheduling problem for the system, and then propose an optimal solution using convex optimization to maximize the overall network throughput. Numerical results show a significant throughput gain achieved by our proposed design over two other baseline schemes.
\end{abstract}

\begin{IEEEkeywords}
Bistatic backscatter, ambient backscatter, RF energy harvesting, IoT, low-power sensor networks.
\end{IEEEkeywords}

\section{Introduction}

\IEEEPARstart{R}{ecently}, wireless-powered communication~\cite{WPCN} has emerged as a promising solution to support power-constrained wireless networks such as wireless sensor networks (WSNs) and Internet-of-Things (IoT). In a wireless powered communication network (WPCN), wireless devices can harvest energy from a dedicated or an ambient FR source, then use such harvested energy for their own data transmission. However, a wireless-powered transmitter may require a long period of time to acquire sufficient energy for active transmissions, and thus the performance of the system is not high. Several approaches have been proposed to deal with this issue such as scheduling time for harvesting and transmitting processes, increasing the transmit power of RF sources, and adopting multiple-input multiple-output (MIMO) technology~\cite{WPCN}. Nevertheless, these solutions are costly, bulky, and not efficient for low-power WSNs. Recently, ambient backscatter communication has been introduced as a cutting-edge technology that enables two wireless nodes to communicate without requiring active RF transmissions~\cite{Huynh2017Survey}. Using ambient backscatter communications, wireless devices can communicate with each other by modulating and reflecting surrounding ambient signals. Thus, this technology can be integrated into WPCNs to improve performance for current WPCNs. However, when being integrated into WPCNs, we need to address the tradeoff problem between backscattering time and energy harvesting time to maximize the overall throughput for the WPCNs.

In~\cite{Ju2014Throughput}, the authors proposed a novel protocol, namely harvest-then-transmit (HTT), to maximize throughput for WPCNs. This protocol will optimize the energy harvesting time in the first phase, and then allocate transmission time for wireless devices based on time-division-multiple-access technique in the second phase. In~\cite{Zhou2017An}, the authors designed a backscatter transmitter selection technique to address multiple access problem in backscatter communication systems. At each time slot, the tag reader will select only one sensor with the best channel condition to perform backscattering. The idea of integrating backscatter communications into RF-powered cognitive radio networks (CRN) was first introduced in~\cite{Hoang2017Ambient}. This idea allows a wireless device to tradeoff between energy harvesting and backscatter time based on their capability and received signal conditions, thereby maximizing its throughput. This work was then extended in~\cite{Hoang2017Optimal}  to maximize the overall throughput for an RF-powered CRN with multiple wireless devices. In both~\cite{Hoang2017Ambient} and~\cite{Hoang2017Optimal}, by incorporating the HTT and backscatter communication techniques, the authors demonstrated that the performance of the secondary system can be improved significantly compared with using either backscatter communication or HTT protocol.

Different from all aforementioned work, in this letter we study the energy and communication efficiency problem for a low-energy communication system with multiple WPCNs co-exiting in the same area. In this system, there is one wireless energy source used to simultaneously supply energy for multiple WPCNs as illustrated in Fig.~\ref{Fig.system_model_new}. In the energy harvesting phase, we can schedule for WPCNs to backscatter and harvest energy alternately to avoid interference among them. Based on the amount of harvested energy, we then can optimize the transmission time for WPCNs to maximize the overall throughput of the network. To obtain the optimal tradeoff between the backscattering, energy harvesting, and data transmission time among WPCNs, we formulate the optimization problem and prove that this problem is concave. Through simulation results, we demonstrate that the proposed solution always achieves the best performance compared with current optimal scheduling mechanisms. 
\section{System Model}
\label{Sec.System}

\subsection{Network Setting} 
We consider a scenario in which there are $N$ WPCNs coexisting in the same area, and they are supplied energy by an RF energy source as shown in Fig.~\ref{Fig.system_model_new}. Each WPCN includes one pair of transmitter and receiver (i.e., sensor and its corresponding gateway as illustrated in Fig.~\ref{Fig.system_model_new}). When the energy source transmits signals, the sensors can harvest energy and store in their batteries to serve for their internal operations as well as data transmissions to their gateways. Alternatively, the sensors can utilize such signals to transmit data through implementing backscatter communication technologies as shown in~\cite{Huynh2017Survey},~\cite{Hoang2017Optimal}, and~\cite{LiuAmbient2013}.
\begin{figure}[tbh]
	\captionsetup{singlelinecheck=off}
	\centering
	\includegraphics[scale=0.25]{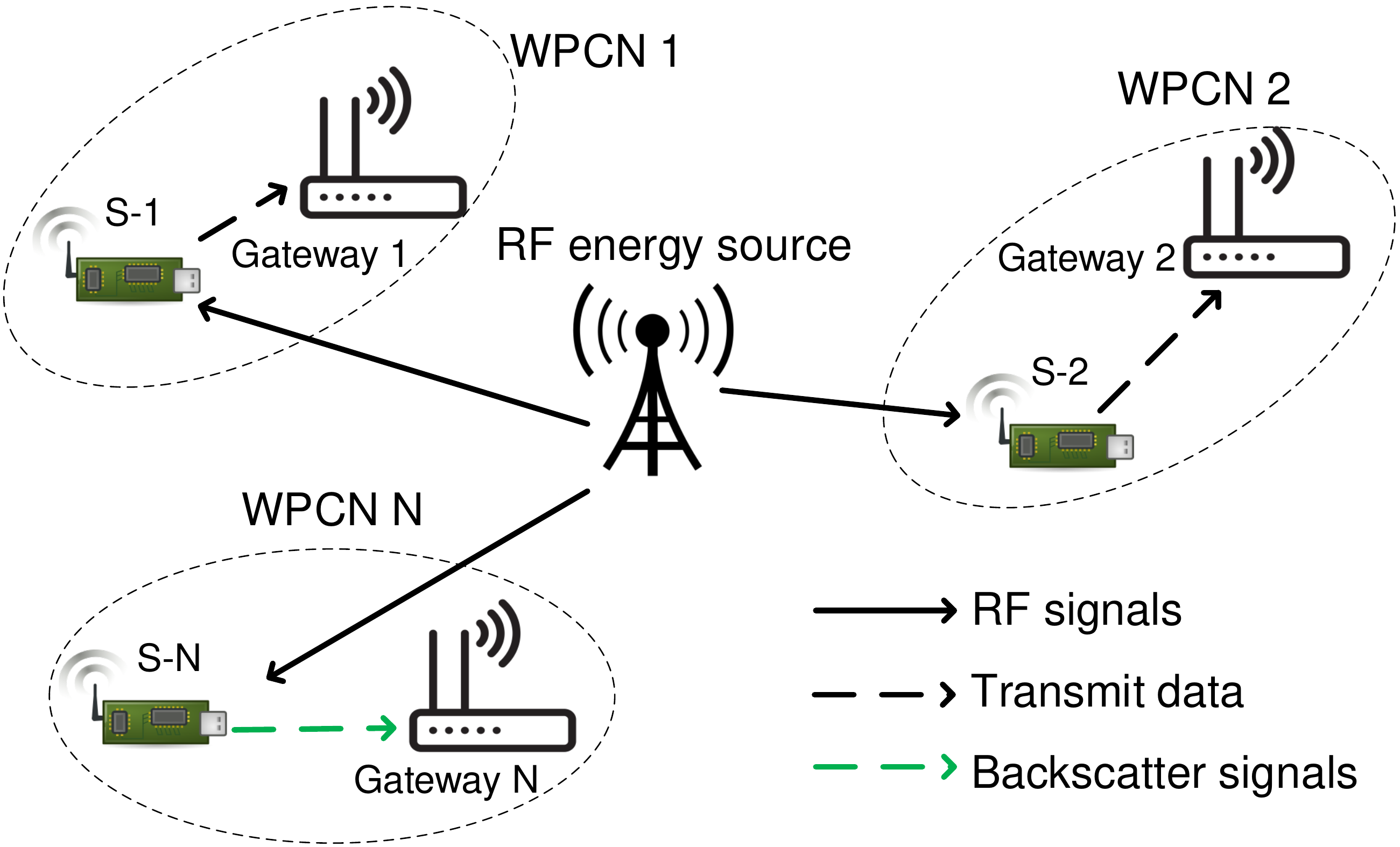}
	\caption{System model.}
	\label{Fig.system_model_new}
\end{figure}
\subsection{Energy Tradeoff and Scheduling}
Similar to~\cite{Ju2014Throughput}, we also consider two successive working phases, i.e., energy harvesting and data transmission phases. However, in the first phase, i.e., energy harvesting phase, the sensors can either backscatter signals to transmit data or harvest energy and use such energy to transmit data in the second phase. Importantly, in our system, when a sensor backscatters signals, all other sensors still can harvest energy as usual. This characteristic allows us to maximize the resource use of RF signals transmitted by the energy source. However, if more than one WPCN backscatter at the same time, they can cause interference for each other. Thus, we need to schedule backscattering time for sensors to avoid such interference. Backscatter time of a sensor is allocated depending on its energy demand, backscatter capability, channel conditions, and the time relation with other sensors to maximize the overall network throughput. Compared with time scheduling problem in [5], our problem is much more complicated because the the energy harvesting and data transmission time are not independent. The more time we allocate for the first phase, the less time we have for sensors to transmit data in the second phase. Furthermore, due to the energy constraint of sensors in low-power sensor networks, we consider the energy optimization problem for sensors more seriously by controlling the amount of harvested energy and transmission power at sensors.

We denote $t^b_n$ as the normalized time period for sensor S-$n$ to backscatter signals, and thus the normalized time period for harvesting energy of sensor S-$n$ is $(\sum_{n=1}^{N}t^b_n)-t^b_n$. The normalized time for data transmission of sensor S-$n$ is denoted as $t^a_n$. Moreover, we denote ${\mathbf{t}}^{\mathrm{b}} = \left[ \begin{array}{ccc} t^b_1 & \cdots & t^b_N \end{array} \right]^\top$ and ${\mathbf{t}}^{\mathrm{a}} = \left[ \begin{array}{ccc}  t^a_1 & \ldots & t^a_N \end{array} \right]^\top$ as the vectors of backscattering time and transmission time of sensors in the network, respectively. Then the following constraints are imposed, i.e., 
\begin{equation}
\begin{aligned}
(\mathbf{C}_0) \phantom{10}		\text{s.t.}		\left\{	\begin{array}{ll}
\sum_{n=1}^{N}t^b_n+\sum_{n=1}^{N}t^a_n \leq 1,		\\
t^b_n, t^a_n \geq 0, \forall n \in \{1, \ldots, N\}.
\end{array}	\right.
\end{aligned}
\label{eq:C0}
\end{equation}
This constraint is to ensure that the time variables are non-negative and the total backscatter and transmission time does not exceed the normalized time frame. 

\section{Problem Formulation} 

We aim to maximize the overall network throughput, i.e., the total number of information bits transmitted by all the sensors in the network per time unit of the network. We denote $R_{\mathrm{sum}}$ as the overall transmission rate which is defined as $R_{\mathrm{sum}} = \sum_{n=1}^{N} R_n = \sum_{n=1}^{N} \big( R_n^{\mathrm{b}} + R_n^{\mathrm{a}} \big)$,
where $R_n^{\mathrm{b}}$ and $R_n^{\mathrm{a}}$ are the numbers of transmitted bits in the backscatter mode and the HTT mode of sensor S-$n$ in one time unit, respectively. 

\subsection{Backscatter Mode} 
\label{subsec:BM}

Let $B^{\mathrm{b}}_n$ denote the transmission rate from using the bistatic backscatter communication of sensor S-$n$. Then, the total number of bits transmitted using the backscatter mode for sensor S-$n$ is expressed as $R^{\mathrm{b}}_n =\eta_n t_n^b B^{\mathrm{b}}_n$,
where $t_n^b$ is the backscattering time and $\eta_n$ is the backscattering efficiency of sensor S-$n$~\cite{Bletsas2008Anti}. It is important to note that when a sensor backscatters signals to the gateway, it still can harvest energy from the RF signals. Although the amount of harvested energy is not enough to transmit data using active wireless transmissions, it is sufficient to sustain backscatter operations of the sensor~\cite{LiuAmbient2013}. Therefore, we do not need to consider the circuit energy consumption for the backscatter mode. 

\subsection{Harvest-then-Transmit Mode} 
\label{subsec:HTTM}

The HTT mode consists of two periods, i.e., harvesting and transmission periods. In the energy harvesting period, the sensors harvest energy from the RF energy source's signals. Then, the sensors use the harvested energy to transmit data in the active transmission period. In the following, we formulate the amount of energy harvested in the first period and the total number of bits transmitted by the sensors in the second period.

\subsubsection{Harvesting Energy} 

As stated in~\cite{Balanis2012}, we can compute the amount of received energy from the RF energy source at sensor S-$n$ in a free space by Friis equation as follows:
\begin{equation}
\label{eq:Friis}
P^{\mathrm{R}}_n = \delta_n P^{\mathrm{T}} \frac{G^{\mathrm{T}} G^{\mathrm{R}}_n \lambda^2}{(4 \pi d_n)^2}	,
\end{equation}
where $P^{\mathrm{R}}_n$ is the received power, $\delta_n$ is the energy harvesting efficiency. $P^{\mathrm{T}}$ is transmit power of the RF energy source, $G^{\mathrm{T}}$ is the antenna gain of the RF energy source, $G^{\mathrm{R}}_n$ is the antenna gain of sensor S-$n$, $\lambda$ is the wavelength, and $d_n$ is the distance between the RF energy source and sensor S-$n$. We then derive the total amount of harvested energy for sensor S-$n$ as follows:
\begin{equation}
\label{eq:amount_EH}
E^{\mathrm{h}}_n = \big((\sum_{n=1}^{N}t^b_n)-t^b_n\big) P^{\mathrm{R}}_n,
\end{equation}
where $\big((\sum_{n=1}^{N}t^b_n)-t^b_n\big)$ is the total energy harvesting time of sensor S-$n$. 

\subsubsection{Transmitting Data} 

After harvesting energy in the first period, sensor S-$n$ uses all the harvested energy to transmit data over $t^a_n$ in the transmission period. According to~\cite{Ju2014Throughput}, we assume that the circuit energy consumption of the sensors is negligible. Let $P^{\mathrm{a}}_n$ denote the transmit power of sensor S-$n$ in the data transmission period $t^a_n$. Therefore, $P^{\mathrm{a}}_n$ can be obtained from $\frac{E^{\mathrm{h}}_n }{t^a_n}$,
where $E^{\mathrm{h}}_n$ is the total amount of harvested energy of sensor S-$n$. From~\cite{Huang2012Decentralized}, given the transmit power $P^{\mathrm{a}}_n$, the transmission rate can be determined as follows:
\begin{equation}
r_n^a = \epsilon_n W \log_2 \left( 1+ \frac{P^{\mathrm{a}}_n}{P_n^0} 	\right)	,
\end{equation}
where $\epsilon_n \in (0,1)$ is the transmission efficiency, $W$ is bandwidth of the channel from the sensor to the gateway, and $P_n^0$ is the ratio between noise power $N_0$ and the channel gain coefficient $g_n$, i.e., $P_n^0=\frac{N_0}{g_n}$. Then, the total number of transmitted bits of sensor S-$n$ using the HTT mode during the transmission time $t^a_n$ is given by:
\begin{equation}
\begin{aligned}
R_n^{\mathrm{a}} & = \psi_n t^a_n \log_2 \left( 1 + \gamma_n \frac{\big((\sum_{n=1}^{N}t^b_n)-t^b_n\big) P^{\mathrm{R}}_n}{t^a_n} 	\right), 
\end{aligned}
\end{equation}
where $\psi_n = \epsilon_n W$ and $\gamma_n = \frac{1}{P_n^0}$. 
The total throughput of sensor S-$n$ is expressed as follows:
\begin{equation}
\begin{aligned}
&	R_n 	 =	t_n^b B^{\mathrm{b}}_n + \psi_n t^a_n \log_2 \left( 1 + \gamma_n \frac{\big((\sum_{n=1}^{N}t^b_n)-t^b_n\big) P^{\mathrm{R}}_n}{t^a_n} 	\right) .	\label{eq:R_n} 
\end{aligned}
\end{equation}
Then, the overall transmission rate of all sensors in the system is given by:
\begin{equation}
\begin{aligned}
& R_{\mathrm{sum}} = \sum_{n=1}^{N} \left[ t_n^b B^{\mathrm{b}}_n + \psi_n t^a_n \log_2 \left( 1 + \gamma_n \frac{\big((\sum_{n=1}^{N}t^b_n)-t^b_n\big) P^{\mathrm{R}}_n}{t^a_n} 	\right)	\right] 	.	\label{eq:R_sum}
\end{aligned}
\end{equation}
In~(\ref{eq:R_sum}), the values of $t^b_n$ and $t^a_n$ must satisfy the constraint $(\mathbf{C}_0)$ in~(\ref{eq:C0}). Additionally, we impose the following constraints:
\begin{equation}
(\mathbf{C}_1) \phantom{10}
\frac{\big((\sum_{n=1}^{N}t^b_n)-t^b_n\big)P^R_n}{t^a_n} \leq P^{\dagger}_n, \phantom{5} \forall n \in \{1, \ldots, N\},
\label{eq:C1}
\end{equation}
\begin{equation}
(\mathbf{C}_2) \phantom{10}
\big((\sum_{n=1}^{N}t^b_n)-t^b_n\big)P^R_n \geq E^0_n, \phantom{5} \forall n \in \{1, \ldots, N\}.
\label{eq:C2}
\end{equation}
The constraint $(\mathbf{C}_1)$  ensures that the transmit power of each sensor must be guaranteed to be lower than or equal to a predefined threshold $P^{\dagger}_n$ due to the power regulation or the power limitation of sensor's circuits. The constraint $(\mathbf{C}_2)$ guarantees that the total amount of harvested energy of sensor S-$n$ must be sufficient for its operations.  
The optimization problem can be formulated as follows:
\begin{equation}
\begin{aligned}
(\mathbf{P}_1) \phantom{10} & \max_{\boldsymbol{t^b}, \boldsymbol{t^a}} 	R_{\mathrm{sum}}, 
& \text{s.t}\quad(\mathbf{C}_0), (\mathbf{C}_1), \textrm{and}
\: (\mathbf{C}_2).
\end{aligned}
\label{eq:OP}
\end{equation}

To find an optimal solution for the optimization problem proposed in~(\ref{eq:OP}), we first prove that the objective function $R_{\mathrm{sum}}$ is a concave function.
\begin{theorem}[]
	\label{theo:convexity}
	The objective function $R_{\mathrm{sum}}$ is a concave function $\forall t_n^b, t^a_n (n \in \mathbb{N})$ satisfying the constraints ($\mathbf{C}_0$), ($\mathbf{C}_1$), and ($\mathbf{C}_2$).
\end{theorem}
\begin{proof}
Due to limited space, we briefly prove Theorem~\ref{theo:convexity}. To prove that $R_{\mathrm{sum}}$ is a concave function, we first prove that $R_n$ is a concave function of $\boldsymbol{\mathrm{t^b}}$ and $\boldsymbol{\mathrm{t^a}}$. 
We can derive the Hessian matrix of the objective function $R_{\mathrm{n}}$ as $\mathbf{H} = \nabla^2 R_{n} (\boldsymbol{t^b}, \boldsymbol{t^a})$. Given an arbitrary real vector $\mathbf{v} = \left[ \begin{array}{ccccccc} v_1 & \ldots & v_n & \ldots & v_{2n}& \ldots & v_{2N} \end{array} \right]^\top$ we have:
\begin{equation}
\begin{aligned}
& \mathbf{v}^\top \mathbf{H} \mathbf{v} 
& = - \frac{\psi_n \big( 	v_n \gamma_n P^{\mathrm{R}}_n(N-1) - v_{2n}w_n	\big)^2}{t^a_n (1+w_n)^2 \ln 2},
\end{aligned}
\end{equation}
where $w_n = \gamma_n \frac{\big((\sum_{n=1}^{N}t^b_n)-t^b_n\big) P^{\mathrm{R}}_n}{t^a_n}$. Since $t^a_n \geq 0, \psi_n>0, \forall n \in \mathbb{N}$, we have $\mathbf{v}^\top \mathbf{H} \mathbf{v} \leq 0, \forall n \in \mathbb{N}$. Thus $R_n(\boldsymbol{t^b}, \boldsymbol{t^a})$ is a concave function. We then can derive that $R_{\mathrm{sum}}$ is a concave function. The proof now is completed. 
\end{proof}

Since the objective function $R_{\mathrm{sum}}$ is a concave function, we adopt the \emph{interior-point method} to find an optimal solution of the optimization problem given in~(\ref{eq:OP}).


\section{Performance Evaluation} 
\begin{figure}[!]
	\captionsetup{singlelinecheck=off}
	\centering
	\includegraphics[scale=0.22]{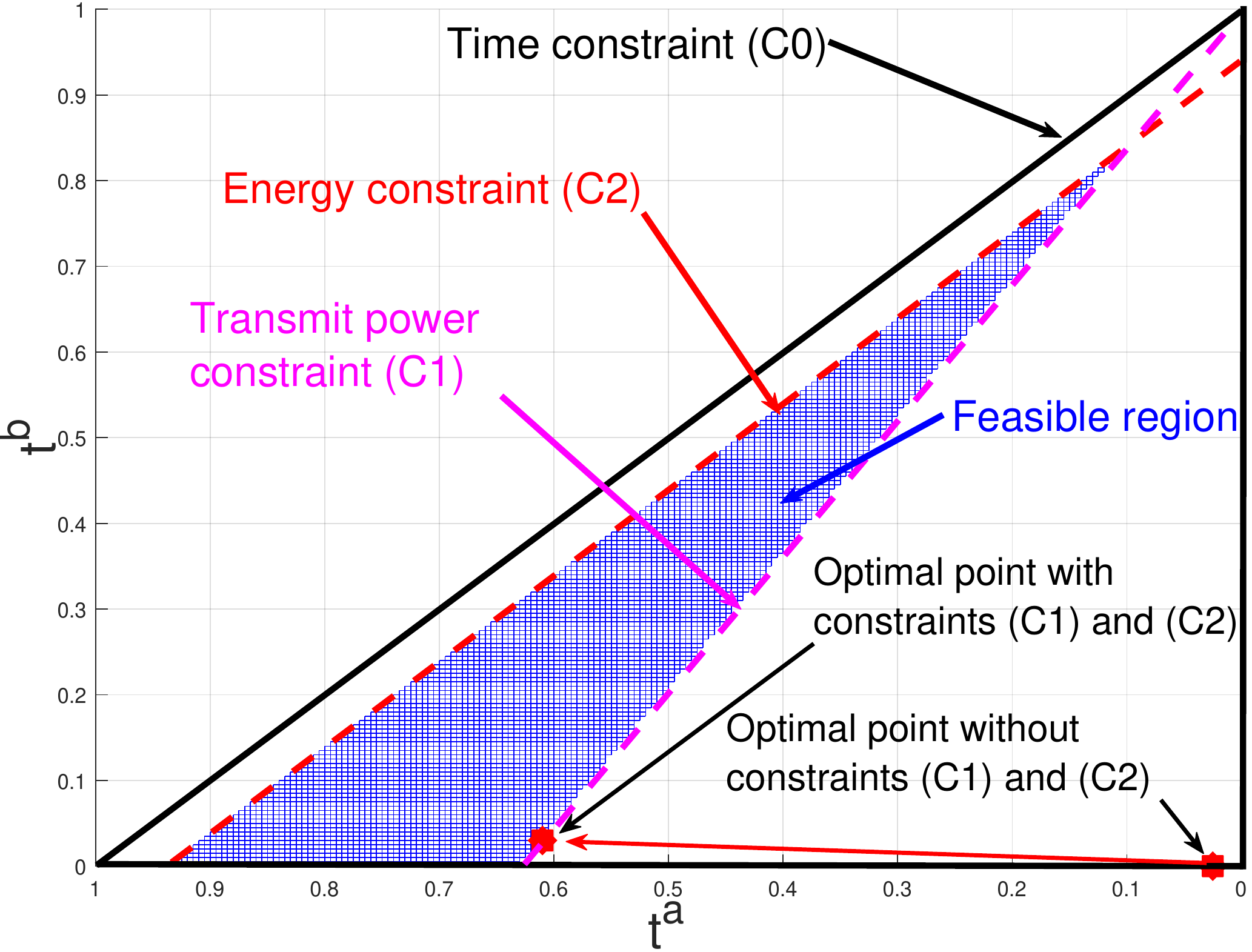}
	\caption{Feasible region and the constraints when $\delta_n$ = 0.6.}
	\label{Fig.objective_function}
\end{figure}
\begin{figure}[!]
	\captionsetup{singlelinecheck=off}
	\centering
	\includegraphics[scale=0.22]{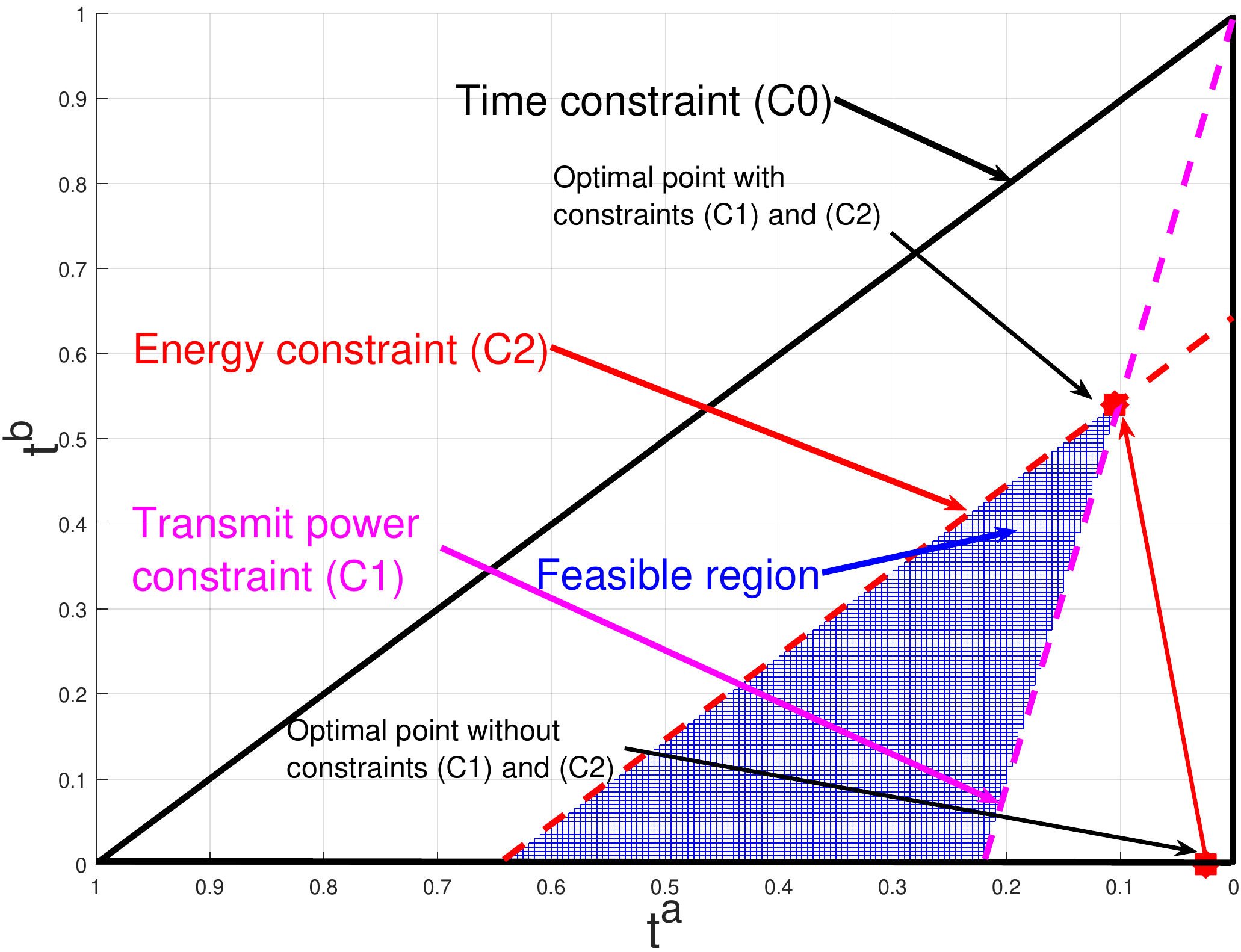}
	\caption{Feasible region and the constraints when $\delta_n$ = 0.1.}
	\label{Fig.objective_function1}
\end{figure}
\begin{figure*}[htbp]
	\centering
	\begin{subfigure}[b]{0.3\textwidth}
		\centering
		\includegraphics[scale=0.2]{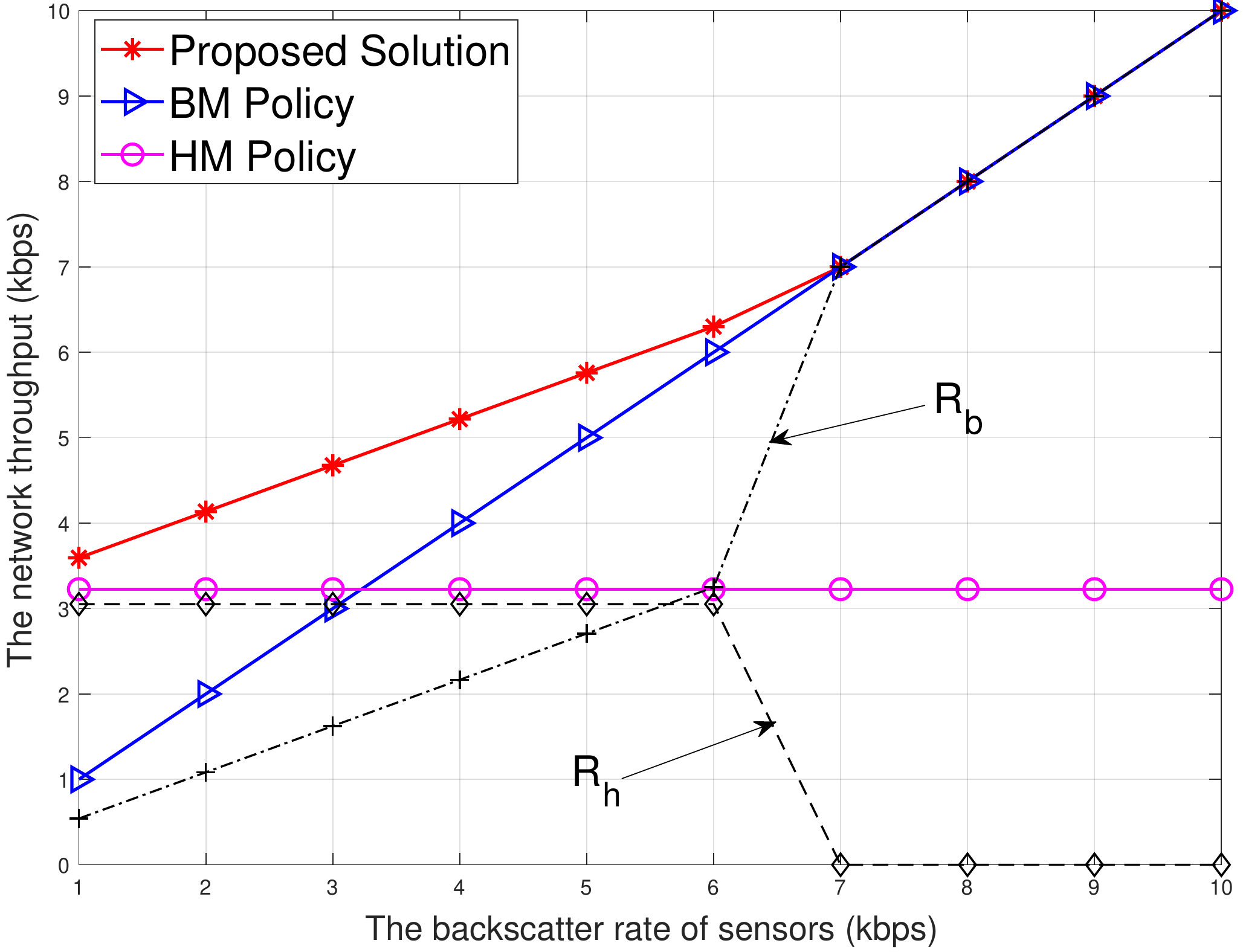}
		\caption{}
	\end{subfigure}%
	~ 
	\begin{subfigure}[b]{0.3\textwidth}
		\centering
		\includegraphics[scale=0.2]{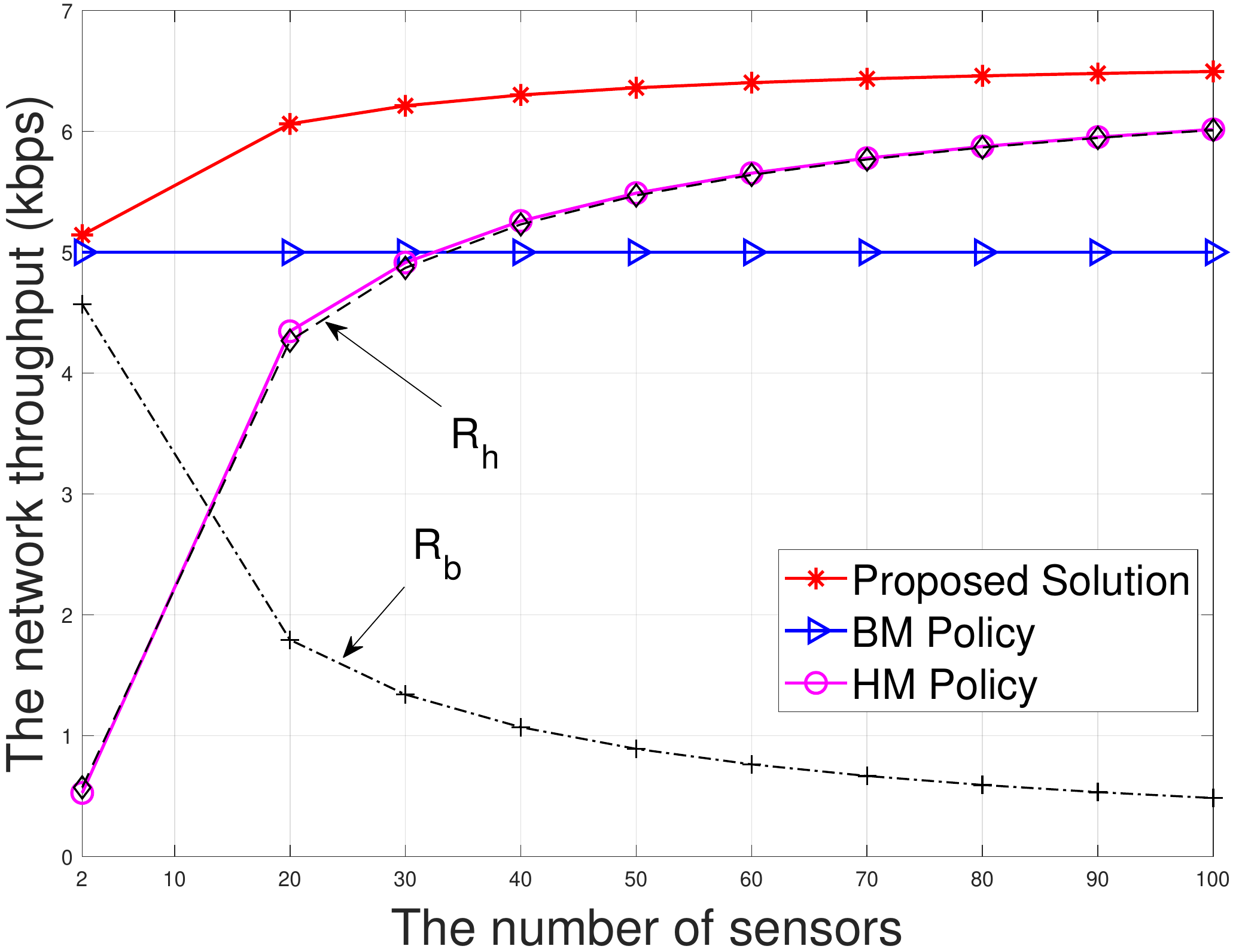}
		\caption{}
	\end{subfigure}%
	~ 
	\begin{subfigure}[b]{0.3\textwidth}
		\centering
		\includegraphics[scale=0.2]{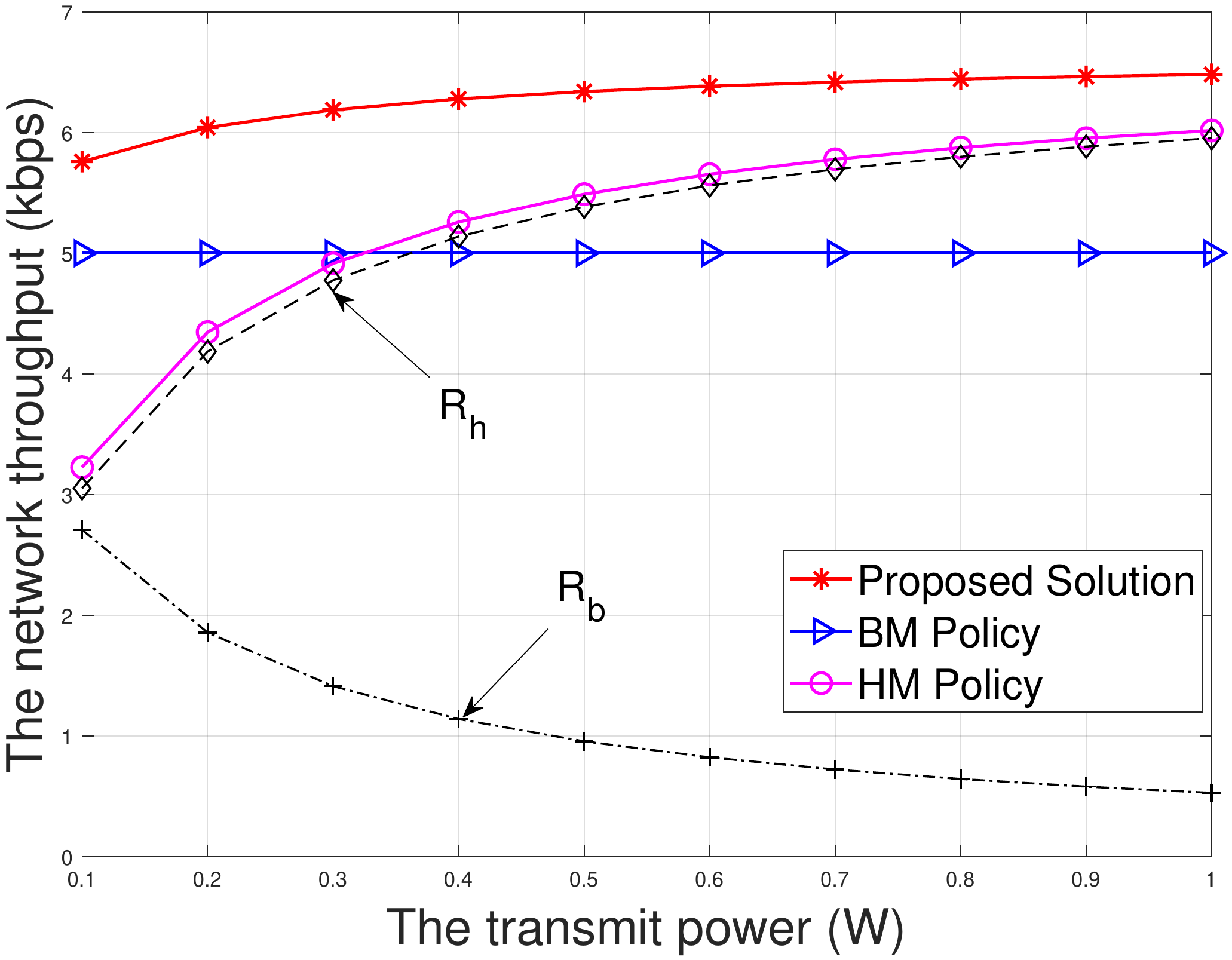}
		\caption{}
	\end{subfigure}
	\caption{The network throughput of the system when (a) the backscatter rate of sensors, (b) the number of sensors, and (c) the transmit power of the RF energy source are varied. $R_b$ and $R_h$ are the network throughput obtained by the backscatter mode and the HTT mode of the proposed solution, respectively.} 
	\label{fig:vary_Backscatter_100_Power}
\end{figure*}
\subsection{Experiment Setup} 
In the system under consideration, the bandwidth and the frequency of the RF signals are 10 MHz and 2.4 GHz, respectively. The RF energy source antenna gain and the sensors antenna gain are set at 6 dBi as in~\cite{Kim2010Reverse}. The distance from the RF energy source to the sensors is assumed to be 10 meters. Unless otherwise stated, the transmit power of the RF energy source is 20 dBm as in~~\cite{Iyer2016Inter}. The energy harvesting efficiency and the data transmission efficiency are set at 0.6. We assume that the sensors are equipped with an RF transceiver CC2420 (\url{http://www.ti.com/product/CC2420}) and its transmit power is -20 dBm as stated in~\cite{Strommer2007Ultralowpower}.  In addition, we compare the performance of the proposed solution with that introduced in~\cite{Ju2014Throughput}, i.e., using HTT protocol only, and with the solution using ambient backscatter communication only~\cite{LiuAmbient2013}. We refer to the former as HTT mode (HM) policy and the latter as backscatter mode (BM) policy.

\subsection{Numerical Results}
We first consider the case with one sensor to show the impact of the constraints on the optimal time allocation policy of the system. We set the transmit power of the RF energy source at 1.8 W, the backscatter rate at 4 kbps, the minimum required power and the maximum transmit power of the sensor is $10^{-6}$ W and $10^{-5}$ W, respectively. As observed in Fig.~\ref{Fig.objective_function}, the optimal time scheduling policy will be changed due to the influence of constraints on feasible region of the solution set. In particular, without energy and transmission power constraints, the system will spend most of the time to harvest energy in the first phase and then use such energy to transmit data  in the second phase. However, under the constraints, the system has to balance among backscatter, energy harvesting, and transmission time to maximize its throughput and satisfy the constraints. In Fig.~\ref{Fig.objective_function1}, when we reduce the energy harvesting efficiency to 0.1 for low received power at the rectifier, the time spending for backscattering increase up to 55\%. This implies that the energy harvesting efficiency is also a critical factor largely affecting the time allocation policy of the system.

We then increase the number of sensors to 10 and evaluate the system performance. The backscatter rates of the sensors are varied from 1 kbps to 10 kbps by changing resistor-capacitor (RC) circuit components~\cite{LiuAmbient2013} in Fig.~\ref{fig:vary_Backscatter_100_Power}(a) and remain at 5 kbps in Fig.~\ref{fig:vary_Backscatter_100_Power}(b) and Fig.~\ref{fig:vary_Backscatter_100_Power}(c). As shown in Fig.~\ref{fig:vary_Backscatter_100_Power}(a), the proposed solution achieves the highest throughput when the backscatter rate increases from 1 kbps to 7 kbps. When the backscatter rate is above 7 kbps, the sensors spend all the time to backscatter their data to the receivers. This means that the proposed solution will switch to the backscatter mode when the backscatter rates of the sensors are high. In Fig.~\ref{fig:vary_Backscatter_100_Power}(b), we vary the number of sensors from 2 to 100 sensors. As shown in Fig.~\ref{fig:vary_Backscatter_100_Power}(b), the throughput of the HTT mode increases. This result stems from the fact that in the energy harvesting phase, multiple sensors can harvest energy simultaneously while they cannot backscatter data at the same time. Thus, the sensors will spend more time to harvest energy instead of backscattering RF signals. Similarly, the overall throughput of the HTT mode also increases in the case when the transmit power of the RF energy source increases as shown in Fig.~\ref{fig:vary_Backscatter_100_Power}(c). The reason is that when the transmit power increases, the sensors can harvest more energy from the RF signals. As a result, the sensors will adopt the HTT mode. Importantly, in all cases, the proposed solution always achieves the best performance as the backscattering time and transmission time are balanced.

\section{Summary} 
In this letter, we have studied the network performance optimization problem for the WPCN with backscatter communications. In this network, the sensors can cooperate to maximize the overall network throughput under energy constraints of low-power wireless sensor networks. To do so, we formulate the performance optimization problem with energy constraints and prove that the problem is concave. Through numerical results, we have shown that our proposed solution always achieves the best performance compared with conventional solutions under different parameter settings.

\end{document}